\documentclass[letter,11pt]{amsart}

\usepackage[margin=1in]{geometry}
\usepackage{graphicx}
\usepackage{amsmath}
\usepackage{amssymb}
\usepackage{amsthm}
\usepackage{enumitem}
\usepackage{lscape}
\usepackage{booktabs}
\usepackage{bm}
\setitemize{itemsep=0pt, topsep=2pt}
\setenumerate{itemsep=0pt, topsep=2pt}
\usepackage{indentfirst}

\newcommand{\Z}{\mathbb{Z}}
\newcommand{\R}{\mathbb{R}}
\newcommand{\C}{\mathbb{C}}
\newcommand{\cF}{\mathcal{F}}

\newcommand{\la}{\left\langle}
\newcommand{\ra}{\right\rangle}
\newcommand{\itpi}{i2\pi}
\newcommand{\dd}{\, \text{d}}
\newcommand{\pd}[2]{\frac{\partial #1}{\partial #2}}
\newcommand{\od}[2]{\frac{\mathrm{d} #1}{\mathrm{d} #2}}
\newcommand{\bra}[1]{\la #1 \right|}
\newcommand{\ket}[1]{\left| #1 \ra}
\newcommand{\modks}[1]{\, (\text{mod } #1)}
\newcommand{\comment}[1]{}

\DeclareMathOperator{\tr}{tr}

\theoremstyle{plain}
\newtheorem{theorem}{Theorem}

\newtheorem{lemma}{Lemma}

\theoremstyle{definition}
\newtheorem{definition}{Definition}
\newtheorem{property}{Property}
\newtheorem*{remark}{Remark}

\newcommand{\comm}[1]{}

\usepackage{xcolor}
\definecolor{cerise}{RGB}{222, 49, 99}

\title{On discrete Wigner transforms}

\author{Zhenning Cai}
\address{Department of Mathematics, National University of
  Singapore, Level 4, Block S17, 10 Lower Kent Ridge Road, Singapore 119076}
\email{matcz@nus.edu.sg}

\author{Jianfeng Lu}
\address{Department of Mathematics, Department of Physics,
  Department of Chemistry, Duke University, Box 90320, Durham NC 27708, USA}
\email{jianfeng@math.duke.edu}

\author{Kevin Stubbs}
\address{Department of Mathematics, Duke University, Box 90320, Durham NC 27708, USA}
\email{kstubbs@math.duke.edu}

\begin{document}

\begin{abstract}
  In this work, we derive a discrete analog of the Wigner transform over the space 
  $(\C^p)^{\otimes N}$ for any prime $p$ and any positive
  integer $N$. We show that the Wigner transform over this space can
  be constructed as the inverse Fourier transform of the standard
  Pauli matrices for $p=2$ or more generally of the Heisenberg-Weyl
  group elements for $p > 2$.  We connect our work to a previous
  construction by Wootters \cite{wootters} of a discrete Wigner
  transform by showing that for all $p$, Wootters' construction
  corresponds to taking the inverse symplectic Fourier transform
  instead of the inverse Fourier transform. Finally, we discuss some
  implications of these results for the numerical simulation of
  many-body quantum spin systems.
\end{abstract}

\maketitle

\section{Introduction}

For functions $f \in \R^d$, it is well known that its Wigner distribution, given by 
\begin{equation}
  W[f](\xi, x) = \int e^{- i 2\pi \xi \cdot p} f(x + \tfrac{1}{2} p) \overline{f(x - \tfrac{1}{2}p)}
  \dd p
\end{equation}
provides a useful characterization of the function. The Wigner
transformation has important applications include the phase space
formulation of quantum mechanics \cite{case2008}, time-frequency
analysis in signal processing \cite{cohen}, semiclassical analysis
\cite{folland, zworski}, just to name a few.  

In this paper, we consider the discrete analog of Wigner transforms;
that is, we aim for phase space representations for vectors in
$(\C^p)^{\otimes N}$, analogous to the continuous case. This natural
generalization has been considered by Wootters \cite{wootters} in
1987, surprisingly much later than the continuous formulation
\cite{Wigner1932, Weyl1927}.  Here we revisit the construction of
discrete Wigner transforms and we propose a natural analog of the
continuous Wigner transform essentially by ``discretizing the
integral'' as a sum over $\Z_p$. This leads to a very natural
extension of the continuous Wigner transform to the discrete state
space. Along the way, we will revisit Wootters' construction of the
discrete Wigner transform. While his proposed discrete Wigner
transform was through a quite indirect link to the continuous analog,
the connection becomes much more clear through our perspective.  In
fact, it turns out that our construction is tightly connected to the
Heisenberg-Weyl group as our discrete Wigner transform can be
represented as an inverse Fourier transform of the 
Heisenberg-Weyl group elements, whereas Wootters'
construction corresponds to an inverse symplectic Fourier transform.
As we will see, taking the inverse Fourier transform (as opposed to
the inverse symplectic Fourier transform) is more natural in some
sense.

One of the motivations of our work comes from recent proposed
numerical methods for quantum many-body spin dynamics based on phase
space representation \cite{Schachenmayer2015, Orioli2017}.  These
works extend the phase space numerical methods for Schr\"odinger
equations (see e.g., the review article \cite{Jin2011}) to the setting
of discrete state space. As the phase space representation is a
natural bridge between the Schr\"odinger equations and their
semiclassical limit, such numerical methods are expected to work well
in the semiclassical regime. Recent work in the physics literature
\cite{Schachenmayer2015, Orioli2017} also demonstrates its success
beyond that. Our study of the discrete Wigner transform is a first
step towards numerical analysis of these phase space numerical
methods, see Section~\ref{sec:dynamics}. A complete numerical analysis
will require mathematical understanding of the quantum entanglement of
the spin dynamics, which will be left for future work.

The rest of the paper is organized as follows. We recall some
preliminaries on spin space, discrete Fourier transforms, and the
Heisenberg-Weyl group in Section~\ref{sec:prelim}. In
Section~\ref{sec:DW_and_HW}, we propose a discrete Wigner transform
motivated from the continuous Wigner transform and illustrate its
connection to the Heisenberg-Weyl group. Wootters' construction of the
discrete Wigner transform will be recalled in
Section~\ref{sec:wootters} and compared with the current construction.
We prove that the proposed discrete Wigner transform
satisfies Wootters' requirements for a Wigner transform in Section~\ref{sec:cond}. In
Section~\ref{sec:dynamics}, we discuss the quantum dynamics in the
phase space representation, and in particular provide a derivation of
a recently proposed numerical method in the physics literature.

\section{Preliminaries}\label{sec:prelim}
\subsection{Properties of the $p^{\text{th}}$ roots of unity}
Throughout this paper, we will always consider a fixed prime $p$ and define $\omega$ be a $p^{\text{th}}$ root of unity (e.g., $\omega = e^{\itpi / p}$). $\omega$ has two important properties which we will make use of throughout our calculations.

\begin{property} \label{prop1}
If $p$ is a prime, then performing addition and multiplication by integers in the exponent of $\omega$ is equivalent to performing those operations modulo $p$.
\end{property}
\begin{proof}[Proof Sketch]
  The key to the property is noticing that $\omega^p = e^{\itpi} =
  1$. Suppose we have an integers $n,m\in \Z$ and can write $n = q_1p
  + r_1$ and $m = q_2 p + r_2$ where $q_1, q_2, r_1, r_2$ are integers. We calculate
  \[
    \omega^{n} = \omega^{q_1 p + r_1} = \omega^{q_1 p} \omega^{r_1} = \omega^{r_1}
  \]
  \[
  \omega^{nm} = \omega^{(q_1 p + r_1)(q_2 p + r_2)} = \omega^{q_1q_2
    p^2} \omega^{(r_2q_1 + r_1q_2)p} \omega^{r_1 r_2} = \omega^{r_1
    r_2}
  \]
  The other properties of modular arithmetic easily follow. 
\end{proof}
The qualification that the arithmetic is done with integers is
critical; some later arguments will become more complicated when we
have to consider $\omega$ raised to a non-integer power. 

\begin{property} \label{prop2}
  If $x \in \mathbb{R}$ then 
  \[
    \sum_{n (p)} \omega^{nx} = 
    \begin{cases} 
        p & x \in \Z \text{ and } x \equiv 0 \modks{p}; \\
        0 & \text{otherwise},
      \end{cases}
  \]
  where the notation $\sum_{n (p)}$ denotes any sum over a set of
  equivalence classes mod $p$.
\end{property}

\begin{proof}
  This follows easily from geometric series. If $x \equiv 0 \modks{p}$ then $\omega^x = 1$ so we get $\sum_{n (p)} \omega^{nx} = \sum_{n=0}^{p-1} 1 = p$. Otherwise, $\omega^x \neq 1$ so we use geometric series to conclude
  \[
    \sum_{n (p)} \omega^{nx} = \sum_{n=0}^{p-1} (\omega^{x})^n  = \frac{1 - \omega^{px}}{1 - \omega^x} = \frac{1 - (\omega^{p})^x}{1 - \omega^x} = 0.
  \]
\end{proof}
Because we will be considering many expressions modulo $p$, it will be
useful to define the Kronecker $\delta$-function modulo $p$ for
vectors $\alpha, \beta \in \Z^N$, $\alpha = (a_1, \dots, a_N)$, $\beta
= (b_1, \dots, b_N)$ as follows:
\[
  \delta^p_{\alpha,\beta} := \begin{cases}
    1 & a_i \equiv b_i \modks{p} \text{ for } i \in \{ 1, \dots, N\} \\
    0 & \text{otherwise}
  \end{cases}
\]

\subsection{A Comment on Arithmetic}
We will make explicit the difference between arithmetic in the set
$[p] := \{ 0, 1, \dots, p-1 \}$ and in the group $\Z_p$. In
particular, unless otherwise specified when we add or multiply two
elements from $[p]$ we use standard arithmetic not modular
arithmetic. Because of this convention, addition and multiplication of
elements from the set $[p]$ is not closed. As we will later see, due
to Property \ref{prop1}, the fact that $[p]$ is not closed under these
operations will not be important. As a more explicit example, when
$p=2$ and $a = b = 1$, 
\[
\begin{array}{l}
  a, b \in [p] \Longrightarrow a + b = 1 + 1 = 2 \\
  \addlinespace[2ex]
  a, b \in \Z_p \Longrightarrow a + b = 1 + 1 \equiv 0 \modks{2} \\
\end{array}
\]

\subsection{Spin Space and the Discrete Fourier transform}
Throughout this paper, we will consider the vector space $(\C^p)^{\otimes N}$ where $p$ is a prime number; this space naturally aries in the study of quantum spin systems. Following the bra-ket notation from quantum mechanics, we will denote any vector $v \in (\C^p)^{\otimes N}$ by $\ket{v}$ and use $\bra{v}$ to denote $\ket{v}$'s vector space dual.

When $N=1$, we will denote the standard basis for $\C^p$ as $\ket{0}, \ket{1}, \cdots, \ket{p-1}$ where $\ket{j}$ is a vector with a 1 in the $(j+1)^{\text{th}}$ position. We will adopt the convention that if $a \in \Z$ but $a \not\in [p]$ then $\ket{a} := \ket{a \modks{p}}$ and similarly for $\bra{a}$. 

For $N > 1$, we will write the standard basis for $(\C^p)^{\otimes N}$ as a Kronecker product $\ket{a_1 a_2 \cdots a_N} := \ket{a_1} \otimes \ket{a_2} \otimes \cdots \otimes \ket{a_N}$ where $a_i \in [p]$. For example, \\
\[
  \begin{array}{lll}
    \textbf{Space} & p = 5,\, N = 1 & p = 2,\, N = 2 \\
    \addlinespace[2ex]
    \textbf{Std. Basis} & \ket{0}, \, \ket{1}, \,\ket{2}, \ket{3}, \ket{4} & \ket{00}, \, \ket{01}, \, \ket{10}, \, \ket{11} \\
    \addlinespace[2ex]
    \textbf{Example} & \ket{2} := \begin{bmatrix} 0 & 0 & 1 & 0 & 0 \end{bmatrix}^{\top} & \ket{10} := \begin{bmatrix} 0 & 0 & 1 & 0  \end{bmatrix}^{\top} \\
  \end{array}
\]
~ \\
For our purposes, it will be useful to recall the inverse Fourier transform and inverse symplectic Fourier transform for vectors over $(\C^p)^{\otimes (2N)}$. Let $\alpha_1, \alpha_2, \beta_1, \beta_2 \in [p]^{N}$ where $\alpha_1 = (a_1, \cdots, a_N)$, $\alpha_2 = (a_{N+1}, \cdots, a_{2N})$ and $\beta_1 = (b_1, \cdots, b_N)$, $\beta_2 = (b_{N+1}, \cdots, b_{2N})$. With this notation, we have the Fourier transform, $\cF$, and symplectic Fourier transform, $\cF_{symp}$, given by:
\begin{equation} \label{eq:ft}
\begin{split} 
  \cF & = \frac{1}{p^N} \sum_{\beta_1} \sum_{\beta_2} \omega^{\alpha_1 \cdot \beta_1 + \alpha_2 \cdot \beta_2} \ket{\alpha_1\alpha_2} \bra{\beta_1 \beta_2} \\
           & := \frac{1}{p^N} \sum_{b_1} \cdots \sum_{b_{2N}} \omega^{a_1 b_1 + \cdots + a_{2N}b_{2N}}  \ket{a_1\cdots a_{2N}} \bra{b_1\cdots b_{2N}},
\end{split}
\end{equation}
and
\begin{equation}
\begin{split} \label{eq:symp_ft}
  \cF_{symp} & = \frac{1}{p^N} \sum_{\beta_1} \sum_{\beta_2} \omega^{\alpha_1 \cdot \beta_2 - \alpha_2 \cdot \beta_1} \ket{\alpha_1\alpha_2} \bra{\beta_1 \beta_2} \\
           & := \frac{1}{p^N} \sum_{b_1} \cdots \sum_{b_{2N}} \omega^{(a_1 b_{N+1} + \cdots + a_N b_{2N}) - (a_{N+1} b_{1} + \cdots + a_{2N} b_{N})} \ket{a_1\cdots a_{2N}} \bra{b_1\cdots b_{2N}}.
\end{split}
\end{equation}
Using Property~\ref{prop2}, it is not hard to verify that the following equations give the inverse Fourier transform and inverse symplectic Fourier transform:
\[
  \cF^{-1} = \frac{1}{p^N} \sum_{\beta_1} \sum_{\beta_2} \omega^{-(\alpha_1 \cdot \beta_1 + \alpha_2 \cdot \beta_2)} \ket{\alpha_1\alpha_2} \bra{\beta_1 \beta_2}; 
\]
\[
  \cF^{-1}_{symp} = \frac{1}{p^N} \sum_{\beta_1} \sum_{\beta_2} \omega^{-(\alpha_1 \cdot \beta_2 - \alpha_2 \cdot \beta_1)} \ket{\alpha_1\alpha_2} \bra{\beta_1 \beta_2}. 
\]

\subsection{The Discrete Heisenberg-Weyl group}
An important group which acts on spin space is the Heisenberg-Weyl group. For $N=1$, the Heisenberg-Weyl group can be defined through a collection of unitary operators $D(a_1,a_2)$, $a_1, a_2 \in [p]$:
\[
  D(a_1, a_2) := \sum_{\ell \in [p]} \omega^{a_2 \ell} \ket{a_1 + \ell} \bra{\ell}
\]
Using Property 1 and recalling our convention for bra-ket notation it is not hard to check that for any $b_1,b_2 \in \Z$, $D(b_1, b_2)$ is still well defined and $D(b_1, b_2) = D(b_1 \modks{p}, b_2 \modks{p})$. 

These operators satisfy a number of nice relationships
(see \cite{Howard2006}). In particular, in what follows we will use
the following two facts:
\begin{align*}
  & D(a_1, a_2) D(b_1, b_2) = \omega^{a_2 b_1} D(a_1 + b_1, a_2 + b_2)\\
  & D(a_1, a_2)^{-1} = D(a_1, a_2)^\dagger = \omega^{a_1 a_2} D(-a_1, -a_2).
\end{align*}

Furthermore, using the definition of the $D(a_1, a_2)$, it is not hard to verify that the collection $\{ D(a_1, a_2) \}$ forms an orthogonal basis of $p \times p$ matrices under the trace inner product (i.e. for all $a_1,a_2,b_1,b_2 \in [p]$, $\tr{(D(a_1, a_2)^\dagger D(b_1, b_2))} = p\,\delta^p_{(a_1,a_2),(b_1,b_2)}$).

\subsection{The Wigner and Fourier-Wigner
  Transforms} \label{sec:cont_wigner}
In Section \ref{sec:DW_and_HW} we will look at at the natural discretizations of the Wigner and Fourier-Wigner transforms.  Here we recall the following definitions of the continuous Wigner transform and the continuous Fourier-Wigner transform; for more discussions on those, see e.g., \cite{folland}.\\

\underline{The Wigner Transform on $f,g$:}
\[
  W[f](\xi,x) := \int e^{-\itpi \xi \cdot p} f(x+\tfrac{1}{2}p) \overline{f(x - \tfrac{1}{2}p)} \dd p
\]

\underline{The Fourier-Wigner Transform on $f,g$:}
\[
  FW[f](p,q) := \int e^{\itpi q \cdot y} f(y+\tfrac{1}{2}p) \overline{f(y - \tfrac{1}{2}p)} \dd y
\]

As one might expect from the name, the Fourier-Wigner transform is indeed the Fourier transform of the Wigner transform \cite{folland}.

\section{The Discrete Wigner transform and the Heisenberg-Weyl
  group} \label{sec:DW_and_HW}
\subsection{Deriving the Discrete Wigner and Discrete Fourier Wigner transforms}
Directly replacing the integrals in the definition of Wigner and Fourier-Wigner transforms as in Section \ref{sec:cont_wigner} with sums over $[p]$ and $e^{\itpi}$ with $\omega$ we can naturally the discretize the Wigner and Fourier-Wigner transforms as linear operators acting on $\C^p$ as follows: \\

\underline{The Discrete Wigner Transform:}
\begin{align*}
  W_{dis}(a_1,a_2) & = \sum_{\ell \in [p]} \omega^{-a_1 \ell} \ket{a_2 + \tfrac{1}{2}\ell} \bra{a_2 - \tfrac{1}{2}\ell} \\
                   & = \sum_{\ell \in [p]} \omega^{-a_1 \ell} \ket{a_2 + 2^{-1}\ell} \bra{a_2 - 2^{-1}\ell}
\end{align*}

\underline{The Discrete Fourier-Wigner Transform:}
\begin{align*}
  FW_{dis}(a_1,a_2) & = \sum_{\ell \in [p]} \omega^{a_2 \ell} \ket{\ell + \tfrac{1}{2} a_1} \bra{\ell - \tfrac{1}{2}a_1} \\
                    & = \sum_{\ell \in [p]} \omega^{a_2 \ell} \ket{\ell + 2^{-1} a_1} \bra{\ell - 2^{-1}a_1}
\end{align*}
Note that for the above equations to make sense, we need $2$ to be invertible modulo $p$ which is not the case when $p=2$. It turns out that the final formulas we derive will still hold for the case $p=2$ in some sense. To stress this point, in the above formulas and henceforth we will write $2^{-1}$ instead of $\tfrac{1}{2}$.

For $p > 2$, we can write both $W_{dis}(a_1, a_2)$ and
$FW_{dis}(a_1, a_2)$ quite compactly in terms of the Heisenberg-Weyl
group. Furthermore, the equations we derive for $p>2$ will allow us to
generalize nicely to the case $p=2$, which will follow afterwards.
\begin{theorem} \label{thm:wigner_hw_connection} For $p > 2$, let $U$
  be the ``flip'' operator
  $U := \sum_{\ell \in [p]} \ket{\ell} \bra{-\ell}$. We can write
  $W_{dis}(a_1,a_2)$ and $FW_{dis}(a_1,a_2)$ in terms of the
  Heisenberg-Weyl group as follows:
  \[
    W_{dis}(a_1,a_2) = \omega^{-2a_1 a_2} D(2a_2, -2a_1) U
  \]
  \[
    FW_{dis}(a_1,a_2) = \omega^{2^{-1} a_1a_2} D(a_1, a_2).
  \]
\end{theorem}
\begin{proof}
  This proof is a straightforward calculation. Beginning with $W_{dis}(a_1,a_2)$
  \begin{align*}
    W_{dis}(a_1,a_2) & = \sum_{\ell} \omega^{-a_1 \ell} \ket{a_2 + 2^{-1}\ell} \bra{a_2 - 2^{-1}\ell} \\
                     & = \sum_{\ell} \omega^{-2a_1 \ell} \ket{a_2 +\ell} \bra{a_2 - \ell} \\
                     & = \sum_{\ell} \omega^{-2a_1 (\ell + a_2)} \ket{2a_2 +\ell} \bra{- \ell} \\
                     & = \omega^{-2a_1 a_2} \sum_{\ell} \omega^{-2a_1 \ell} \ket{2a_2 +\ell} \bra{- \ell}.
  \end{align*}
  Recalling our definition for $U$ above we conclude that 
  \begin{align*}
    W_{dis}(a_1,a_2) & = \omega^{-2a_1 a_2} \left(\sum_{\ell} \omega^{-2a_1 \ell} \ket{2a_2 +\ell} \bra{\ell}\right) U\\ 
                     & = \omega^{-2a_1 a_2} D(2a_2, -2a_1) U.
  \end{align*}

  Similarly, for $FW_{dis}(a_1,a_2)$
  \begin{align*}
    FW_{dis}(a_1,a_2) & = \sum_{\ell} \omega^{a_2 \ell} \ket{\ell + 2^{-1} a_1} \bra{\ell - 2^{-1} a_1} \\
                     & = \sum_{\ell} \omega^{a_2 (\ell + 2^{-1} a_1)} \ket{\ell + a_1} \bra{\ell}  \\
                     & = \omega^{2^{-1} a_1a_2} \sum_{\ell} \omega^{a_2 \ell} \ket{\ell + a_1} \bra{\ell} \\
                     & = \omega^{2^{-1} a_1a_2} D(a_1, a_2).
  \end{align*}
\end{proof}

Similar to the continuous case, for $p > 2$, the discrete Wigner transform and the discrete Fourier-Wigner transform are related via the Fourier transform. 

\begin{theorem} \label{inverse_ft_wigner}
  For $p>2$, $W_{dis}(a_1,a_2)$ is the inverse Fourier transform of $FW_{dis}(a_1,a_2)$.
\end{theorem}
\begin{proof}
  The key in proving this theorem is the expansion of $U$ in the
  Heisenberg-Weyl basis. Using the properties of the Heisenberg-Weyl
  group we have
  \begin{align*}
    D(a,b)^\dagger U & = \omega^{ab} D(-a,-b) U \\
               & = \omega^{ab} \left(\sum_{\ell} \omega^{-bc} \ket{-a + \ell}\bra{\ell} \right) U \\
               & = \omega^{ab} \sum_{\ell} \omega^{-bc} \ket{-a + \ell}\bra{-\ell}
  \end{align*}
  Taking the trace of both sides and noting $-a + \ell = -\ell$
  $\Longrightarrow$ $\ell = 2^{-1} a$, we conclude that
  $\tr(D(a,b)^\dagger U) = \omega^{2^{-1}ab}$ and so since the collection
  $\{ D(b_1,b_2) \}$ forms an orthogonal basis:
  \[
    U = \frac{1}{p}\sum_{b_1,b_2} \omega^{2^{-1}b_1b_2} D(b_1, b_2).
  \]
  Now, substituting this expansion into our expression for $W_{dis}(2^{-1}a_1, 2^{-1}a_2)$ we get:
  \begin{align*}
    W_{dis}(2^{-1}a_1, 2^{-1}a_2) & = \omega^{-2^{-1}a_1 a_2} D(a_2, -a_1) U \\
                                  & = \frac{1}{p}\,\omega^{-2^{-1}a_1 a_2} \sum_{b_1,b_2} \omega^{2^{-1}b_1b_2} D(a_2, -a_1) D(b_1, b_2) \\
                                  & = \frac{1}{p}\,\omega^{-2^{-1}a_1 a_2} \sum_{b_1,b_2} \omega^{2^{-1}b_1b_2} \omega^{-a_1b_1} D(a_2 + b_1, -a_1 + b_2) \\
                                  & = \frac{1}{p}\,\omega^{-2^{-1}a_1 a_2} \sum_{b_1,b_2} \omega^{2^{-1}(b_1-a_2)(b_2+a_1)} \omega^{-a_1(b_1-a_2)} D(b_1, b_2) \\
                                  & = \frac{1}{p}\,\omega^{-2^{-1}a_1 a_2} \sum_{b_1,b_2} \omega^{2^{-1}(b_1-a_2)(b_2-a_1)} D(b_1, b_2) \\
                                  & = \frac{1}{p}\,\omega^{-2^{-1}a_1 a_2} \sum_{b_1,b_2} \omega^{2^{-1} (b_1 b_2 - a_1 b_1 - a_2 b_2 + a_1 a_2)} D(b_1, b_2) \\
                                  & = \frac{1}{p} \sum_{b_1,b_2} \omega^{2^{-1} (-a_1 b_1 - a_2 b_2)} \omega^{2^{-1}b_1 b_2}  D(b_1, b_2).
  \end{align*}
  So we finally get
  \begin{align*}
    W_{dis}(a_1, a_2) & = \frac{1}{p} \sum_{b_1,b_2} \omega^{-(a_1 b_1 + a_2 b_2)} \omega^{2^{-1}b_1 b_2}  D(b_1, b_2) \\
                      & = \frac{1}{p} \sum_{b_1,b_2} \omega^{-(a_1 b_1 + a_2 b_2)} FW_{dis}(b_1,b_2).
  \end{align*}
\end{proof}
Let us stress that the above calculations only make sense when $p>2$. Despite this, we can give meaning to the discrete Fourier-Wigner transform in the case when $p=2$ by interpreting $2^{-1}$ to actually be the real number $\tfrac{1}{2}$. In this case we get the standard Pauli matrices $I,X,Y,Z$:
\begin{equation}
  \begin{array}{ccc}
    FW_{dis}(0,0) = D(0,0) = I & & FW_{dis}(0,1) = D(0,1) = Z \\
                              & & \\
    FW_{dis}(1,0)= D(1,0) = X & &  FW_{dis}(1,1) = i D(1,1) = Y
  \end{array}
\end{equation}
Likewise, we can define the discrete Wigner transform for $p=2$ as the inverse Fourier transform of $FW_{dis}$. More explicitly:
\begin{equation} \label{dis_wigner_p2}
  \begin{array}{ccc} 
    W_{dis}(0,0) = \frac{1}{2} \bigl( I + Z + X + Y \bigr) & & W_{dis}(0,1) = \frac{1}{2} \bigl( I - Z + X - Y \bigr) \\
                              & & \\
    W_{dis}(1,0)= \frac{1}{2} \bigl( I + Z - X - Y \bigr) & &  W_{dis}(1,1) = \frac{1}{2} \bigl( I - Z - X + Y \bigr)
  \end{array}
\end{equation}
It turns out that this choice for $W_{dis}$ and $FW_{dis}$ is in fact the natural choice.

\section{Wootters' Discrete Wigner Transform}\label{sec:wootters}
An alternate approach to defining the Wigner transform is by looking at the key properties of the continuous Wigner transform and taking their discrete analogues. This is the approach taken by Wootters in \cite{wootters}. To define the discrete Wigner transform in \cite{wootters}, Wootters associates a matrix $A_\alpha$ with each point $\alpha = (a_1, a_2) \in \Z_p^2$ and requires the collection $\{ A_\alpha \}_{\alpha \in \Z_p^2}$ to satisfy the following three key properties:

\begin{enumerate}[label=(W\arabic*),leftmargin=50pt]
\item\label{W1} For each $\alpha \in \Z^2_p$, $\tr A_\alpha = 1$;
\item\label{W2} For any $\alpha,\beta \in \Z^2_p$, $\tr A_\alpha^\dagger A_\beta = p \delta^p_{\alpha, \beta}$;
\item\label{W3} The ``line condition'' (see Definition~\ref{defn:line} below).
\end{enumerate}
The first two properties can be thought of as fixing the normalization and requiring the pairwise orthogonality of $\{ A_\alpha \}$ respectively. The third condition (which we have referred to as the line condition) corresponds to the property of the Wigner transform of a quantum state where integrating over infinite strips of the form $S_{a,b,c_1,c_2} := \{ (p,q) \in \R^2 : c_1 \leq ap + bq \leq c_2\}$ gives the probability that certain observables of that state have value between $c_1$ and $c_2$. As one might expect, the corresponding discrete analog of this property involves the definition of parallel lines on a finite lattice.

\subsection{Parallel Lines on $\Z_p^2$ and the Line Condition}
\begin{definition}
  For $(n,m) \in \Z_p^2 \setminus \{(0,0)\}$ and $k \in \Z_p$, the \textbf{line} $\ell_{n,m,k}$ on the lattice $\Z_p^2$ is a set of points:
  \[
    \ell_{n,m,k} := \{ (x,y) : nx + my \equiv k \modks{p} \}.
  \]
\end{definition}
\begin{definition}
  For $(n,m) \in \Z_p^2 \setminus \{(0,0)\}$ fixed, a \textbf{complete set of parallel lines} $L$ is the collection of $p$ lines:
  \[
    L = \{ \ell_{n,m,k} : k \in \Z_p\}
  \]
\end{definition}
It is not hard to show that $\ell_{n,m,k} \cap \ell_{n,m,k'} = \emptyset$ for $k \neq k'$ so each complete set of parallel lines forms a partition $\Z_p^2$. We are now able to define the ``line condition'' \ref{W3}

\begin{definition}[The Line Condition]\label{defn:line}
  
  For any collection of matrices $\{ A_\alpha \}_{\alpha \in \Z_p^2}$ and any line $\ell$ we can define an associated operator $P_\ell$:
  \[
    P_\ell := \frac{1}{p} \sum_{\alpha \in \ell} A_\alpha.
  \]
  We say the collection $\{ A_\alpha \}_{\alpha \in \Z_p^2}$ satisfies
  the line condition if for any complete set of parallel lines, $L$,
  the collection $\{ P_\ell \}_{\ell \in L}$ is a family of orthogonal
  projections which forms a resolution of the identity, i.e., 
  \begin{enumerate}
  \item For all $\ell_{k_1}, \ell_{k_2} \in L$, $P_{\ell_{k_1}} P_{\ell_{k_2}} = P_{\ell_{k_1}} \delta^p_{k_1,k_2}$
  \item $\sum_{\ell \in L} P_\ell = I$.
  \end{enumerate}
\end{definition}

\subsection{Wootters' Discrete Wigner Collection, $A_\alpha$} 
In the paper \cite{wootters}, Wootters gives the following definition for $\{ A_\alpha \}_{\alpha \in \Z_p^2}$, which he claims satisfies \ref{W1}, \ref{W2}, \ref{W3}:
\[
  \begin{array}{cl}
    p=2: & A_\alpha = \frac{1}{2} \left( I + (-1)^{a_1} Z + (-1)^{a_2} X + (-1)^{a_1 + a_2} Y \right) \\
         &\\
    p>2: & (A_{\alpha})_{jk} = \delta^p_{2a_1,j+k} ~\omega^{a_2(j-k)}.
  \end{array}
\]
The following theorem shows the
connection between $W_{dis}$, $FW_{dis}$, and $A_\alpha$.

\begin{theorem} \label{thm:wootter_hw_connection}
  For all $p$, $A_\alpha$ is the inverse symplectic Fourier transform of $FW_{dis}(a_1,a_2)$. In particular, for all $p$, $A_\alpha = W_{dis}(a_2, -a_1)$.
\end{theorem}

\begin{remark}
  This theorem has an exact analog in the continuous
  case in the following sense: If $W[f](\xi, x)$ and $FW[f](p,q)$
  are the Wigner and Fourier-Wigner transforms of $f$
  respectively then the inverse symplectic Fourier transform of
  $FW[f](p,q)$ is $W[f](x, -\xi)$.
\end{remark}

\begin{proof}
  \underline{Case 1, $p=2$:} Notice that in this case $A_\alpha$ is given by a linear combination of the Pauli matrices. With slight abuse of notation, we can specify this linear combination using a matrix equation as follows:
  \[
    \begin{bmatrix}
      A_{00} \\
      A_{01} \\
      A_{10} \\
      A_{11}
    \end{bmatrix}
     =
    \frac{1}{2}
    \begin{bmatrix}
      1 & 1 & 1 & 1 \\
      1 & 1 & -1 & -1 \\
      1 & -1 & 1 & -1 \\
      1 & -1 & -1 & 1
    \end{bmatrix}
    \begin{bmatrix}
      I \\
      Z \\
      X \\
      Y
    \end{bmatrix}
    =
    \frac{1}{2}
    \begin{bmatrix}
      1 & 1 & 1 & 1 \\
      1 & 1 & -1 & -1 \\
      1 & -1 & 1 & -1 \\
      1 & -1 & -1 & 1
    \end{bmatrix}
    \begin{bmatrix}
      FW_{dis}(0,0) \\
      FW_{dis}(0,1) \\
      FW_{dis}(1,0) \\
      FW_{dis}(1,1) 
    \end{bmatrix}\,.
  \]
  Recalling the definition of a $p^2 \times p^2$ inverse symplectic Fourier transform over $\Z_p$ (where $\beta := (b_1, b_2)$) 
  \[
    \cF^{-1}_{symp} = \frac{1}{p} \sum_{\alpha,\beta \in [p]^2} \omega^{-(a_1 b_2 - a_2 b_1)} \ket{\beta} \bra{\alpha}.
  \]
  A simple calculation in the case when $p=2$ and $\omega = e^{\itpi/2} = -1$ verifies the claim. \\
  
  \underline{Case 2, $p>2$:} We will prove this fact by showing that
  \[
    A_\alpha = \omega^{2a_1a_2} D(2a_1,2a_2)U,
  \]
  where $U$ is the ``flip'' operator $U := \sum_{\ell} \ket{\ell}\bra{-\ell}$. Since we have already shown that $W_{dis}(a_1,a_2) = \omega^{-2a_1a_2} D(2a_2,-2a_1)U$ (see Theorem \ref{thm:wigner_hw_connection}) the result follows since $W_{dis}$ is the Fourier transform of $FW_{dis}$. Recall the definition of $A_\alpha$ for $p>2$.
  \[
    (A_{\alpha})_{k\ell} = \delta^p_{2a_1,k+\ell} ~\omega^{a_2(k-\ell)}.
  \]
  Clearly, an entry of $A_\alpha$ is not zero if and only if $2a_1 \equiv k + \ell \Longrightarrow k \equiv 2a_1 - \ell$. Therefore,
  \begin{align*}
    A_\alpha & = \sum_{\ell \in [p]} \omega^{a_2(2a_1 - \ell-\ell)} \ket{2a_1 - \ell} \bra{\ell} \\
             & = \omega^{2a_1a_2} \sum_{\ell} \omega^{-2a_2\ell} \ket{2a_1 - \ell} \bra{\ell} \\
             & = \omega^{2a_1a_2} \sum_{\ell} \omega^{2a_2\ell} \ket{2a_1 + \ell} \bra{-\ell} \\
             & = \omega^{2a_1a_2} \left(\sum_{\ell} \omega^{2a_2\ell} \ket{2a_1 + \ell} \bra{\ell}\right)U \\
             & = \omega^{2a_1 a_2} D(2a_1, 2a_2) U.
  \end{align*}
  Therefore, $A_\alpha = W(a_2, -a_1)$ and we can write
  \[
    A_\alpha = \frac{1}{p} \sum_{b_1,b_2} \omega^{-(a_1 b_2 - a_2 b_1)} FW_{dis}(b_1,b_2).
  \]
\end{proof}

\section{$W_{dis}(a_1,a_2)$ satisfies Wootters' conditions}\label{sec:cond}
In this section, we show that the discrete Wigner transform $W_{dis}$
satisfies Wootters' three conditions, thus verifying from another
angle those are natural generalization of the Wigner transform to
the discrete setting.

\subsection{Properties of $W_{dis}$ and $FW_{dis}$} \label{sec:props_dw_dfw}
We begin by deriving some of the important properties of $W_{dis}$ and $FW_{dis}$
\begin{lemma}
  For all $p$, the collection $\{ FW_{dis}(a_1,a_2) \}$ is pairwise orthogonal under the trace inner product. That is
  \[
    \tr\bigl(FW_{dis}(a_1,a_2)^\dagger FW_{dis}(b_1,b_2)\bigr) = p\,\delta^p_{(a_1,a_2),(b_1,b_2)}
  \]
\end{lemma}
\begin{proof}
  To simplify the proof, we split into two cases, $p=2$ and $p>2$.
  
  \underline{Case 1, $p=2$:} The collection $\{ FW_{dis}(a_1,a_2) \}$ is the Pauli matrices so it can be easily checked that they are pairwise orthogonal.

  \underline{Case 2, $p>2$:} In this case we have
  \begin{align*}
    \tr\bigl(FW_{dis}(a_1,a_2)^\dagger FW_{dis}(b_1,b_2)\bigr) & = \tr\biggl(\omega^{-2^{-1}a_1 a_2}  D(a_1, a_2)^\dagger \omega^{2^{-1}b_1 b_2}  D(b_1, b_2) \biggr) \\
                                                               & = \omega^{2^{-1}(b_1 b_2 - a_1 a_2)}\, \omega^{a_1a_2} \tr\biggl(D(-a_1, -a_2)  D(b_1, b_2) \biggr) \\
                                                               & = \omega^{2^{-1}(b_1 b_2 - a_1 a_2)} \omega^{a_1a_2-b_1 a_2} \tr\biggl(D(b_1-a_1, b_2-a_2) \biggr) 
  \end{align*}
  But $\tr\bigl(D(b_1-a_1, b_2-a_2) \bigr) \neq 0$ if and only if $a_1 \equiv b_1 \modks{p}$ and $a_2 \equiv b_2 \modks{p}$. Since for $p > 2$, $2^{-1}$ is an integer, we conclude that
  \[
    \tr\biggl(FW_{dis}(a_1,a_2)^\dagger FW_{dis}(b_1,b_2)\biggr) = p \omega^{2^{-1}(a_1 a_2 - a_1 a_2)} \omega^{a_1a_2-a_1 a_2} = p.
  \]
\end{proof}

\begin{lemma}
  For all $p$, the collection $\{ W_{dis}(a_1,a_2) \}$ is pairwise orthogonal under the trace inner product. That is
  \[
    \tr\bigl(W_{dis}(a_1,a_2)^\dagger FW_{dis}(b_1,b_2)\bigr) = p \,\delta^p_{(a_1,a_2),(b_1,b_2)}
  \]
\end{lemma}
\begin{proof}
\begin{align*}
  \tr \biggl(W_{dis}(\tilde{a}_1, \tilde{a}_2)^\dagger W_{dis}(a_1, a_2)\biggr) & = \frac{1}{p^2} \sum_{b_1, b_2} \sum_{\tilde{b}_1, \tilde{b}_2} \omega^{(\tilde{a}_1 \tilde{b}_1 + \tilde{a}_2 \tilde{b}_2)}\omega^{-(a_1 b_1 + a_2 b_2)}  \tr\biggl(FW_{dis}(\tilde{b}_1, \tilde{b}_2)^\dagger FW_{dis}(b_1, b_2)\biggr) \\
                                                                               & = \frac{1}{p} \sum_{b_1, b_2} \omega^{(\tilde{a}_1 b_1 + \tilde{a}_2 b_2)} \omega^{-(a_1 b_1 + a_2 b_2)} \\
   & = \frac{1}{p} \sum_{b_1} \omega^{(a_1 - \tilde{a}_1) b_1} \sum_{b_2} \omega^{(a_2 - \tilde{a}_2) b_2}
\end{align*}
So by Property \ref{prop2}, we conclude the result.
\end{proof}

\begin{lemma}
For all $p$ and $a_1,a_2 \in \Z_p$, $W_{\text{dis}}(a_1,a_2)$ is Hermitian.
\end{lemma}
\begin{proof}
  To simplify the proof, we split into two cases, $p=2$ and $p>2$.

  \underline{Case 1, $p=2$:} By definition \eqref{dis_wigner_p2}, we can easily see that $W_{dis}(a_1,a_2)$ is Hermitian for $p=2$.

  \underline{Case 2, $p>2$:} In this case, we first prove the following useful identity which holds for $p>2$:
  \[
    D(a_1,a_2) U = U D(-a_1,-a_2), \text{ where } U := \sum_{\ell \in [p]} \ket{-\ell} \bra{\ell}.
  \]
  We calculate
  \begin{align*}
    D(a_1,a_2) U & = \sum_{\ell} \omega^{a_2\ell} \ket{a_1 + \ell} \bra{-\ell} \\
                 & = \sum_{\ell} \omega^{-a_2\ell} \ket{a_1 - \ell} \bra{\ell} \\
                 & = \sum_{\ell} \omega^{-a_2\ell} \ket{-(-a_1 + \ell)} \bra{\ell} \\
                 & = U D(-a_1,-a_2).
  \end{align*}
  Since by Theorem \ref{thm:wigner_hw_connection} we have that $W_{dis}(a_1,a_2) = \omega^{-2a_1a_2} D(2a_2,-2a_1)U$ and $U$ is obviously Hermitian we have that
  \begin{align*}
    W_{dis}(a_1, a_2)^{\dagger} & = \omega^{2a_1a_2} U^{\dagger} D(2a_2, -2a_1)^{\dagger} \\
                                & = \omega^{-2a_1a_2} U D(-2a_2, 2a_1) \\
                                &  = \omega^{-2a_1a_2} D(2a_2, -2a_1) U \\
                                & = W_{dis}(a_1, a_2).
  \end{align*}
\end{proof}

\subsection{Verifying \ref{W1}, \ref{W2}, \ref{W3}}

\underline{Verifying \ref{W1}:} We calculate
\begin{align*}
  \tr\left( W_{dis}(a_1, a_2)\right)  & = \tr \Bigl(\frac{1}{p} \sum_{b_1, b_2}\, \omega^{-(a_1 b_1 + a_2 b_2)} FW_{dis}(b_1,b_2) \Bigr) \\
                                      & = \frac{1}{p} \sum_{b_1, b_2}\, \omega^{-(a_1 b_1 + a_2 b_2)} \tr\left(FW_{dis}(b_1, b_2) \right) \\
                                      & = 1.
\end{align*}

\underline{Verifying \ref{W2}:} This was already verified in Section \ref{sec:props_dw_dfw}. \\

\underline{Verifying \ref{W3}:} For some $(n,m) \in \Z_p^2 \setminus \{ (0, 0) \}$ fixed, recall the definition of a line on the $\Z_p^2$ lattice
\[
  \ell_{k} := \ell_{n,m,k} = \{ (x,y) : mx + ny \equiv k \modks{p} \}
\]
and a complete set of parallel lines is the collection $\{ \ell_{k} \}_{k \in \Z_p}$. For this line, we have the associated line operator $P_{\ell_k}$ is the following:
\begin{align*}
  P_{\ell_k} & := \frac{1}{p}\sum_{(x,y) \in \ell_k} W_{dis}(x,y) \\
   & = \frac{1}{p^2} \sum_{(x,y) \in \ell_k} \sum_{b_1, b_2}\, \omega^{-(x b_1 + y b_2)} \omega^{2^{-1} b_1 b_2} D(b_1, b_2).
\end{align*}
To make the following arguments simpler, we will split into two cases
$p=2$ and $p>2$.

\underline{Case 1, $p=2$:} When $p=2$ there are only three complete
sets of parallel lines each with two lines. Since have already
explicitly written down $W_{dis}$ for $p=2$ (see Equation \ref{dis_wigner_p2}) we can
also write down the line projections as follows:
\[
\renewcommand{\arraystretch}{1.4}
\begin{array}{c|c|c}
(n,m) & P_{\ell_0} & P_{\ell_1} \\
\hline
  (0,1) & \tfrac{1}{2}(I + X) & \tfrac{1}{2}(I - X) \\
  (1,0) & \tfrac{1}{2}(I + Z) & \tfrac{1}{2}(I - Z) \\
  (1,1) & \tfrac{1}{2}(I + Y) & \tfrac{1}{2}(I - Y) \\
\end{array}
\]
It is an easy calculation to show that for each $(n,m)$ the collection
$\{ P_{\ell_0}, P_{\ell_1} \}$ forms a resolution of the identity.

\underline{Case 2, $p>2$:} 
In what follows, we will assume without loss of generality that
$n \neq 0$ because of the symmetry between $(x,b_1)$ and $(y, b_2)$. Since
$n \neq 0$, if $(x,y)$ is on the line $\ell_k$ then
$y = n^{-1} k - n^{-1} m x$. For notational ease, we define
$K := n^{-1} k$ and $M := n^{-1} m$ and therefore we can write:
\[
  P_{\ell_k} = \frac{1}{p^2} \sum_{x} \sum_{b_1, b_2}\, \omega^{-(x b_1 + (K - M x) b_2)} \omega^{2^{-1} b_1 b_2} D(b_1, b_2).
\]
Collecting all of the terms containing $x$, we get:
\[
  P_{\ell_k} = \frac{1}{p^2}  \sum_{b_1, b_2}\, \omega^{-K b_2} \omega^{2^{-1} b_1 b_2} D(b_1, b_2) \sum_{x} \omega^{-(b_1 - M b_2) x}
\]
Now the sum over $x$ is over a complete set of equivalence classes
modulo $p$. Therefore, by Property \ref{prop2}, that sum is non-zero 
if and only if $b_1 - M b_2 \equiv 0 \modks{p} \Longrightarrow b_1 \equiv M
b_2 \modks
{p}$. Therefore, we can simplify the above as
\[
  P_{\ell_k} = \frac{1}{p} \sum_{b_2}\, \omega^{-K b_2} \omega^{2^{-1} M b_2^2} D(M b_2, b_2).
\]
Now let us show that the collection $\{ P_{\ell_k} \}_{k \in \Z_p}$ is a resolution of the identity. \\

\noindent\underline{$P_{\ell_k}$ are orthogonal projectors:}
In what follows we slightly abuse notation and define $\omega(x) := \omega^x$. With this new notation we calculate
\begin{align*}
  P_{\ell_k} P_{\ell_{\tilde{k}}} & = \frac{1}{p^2} \sum_{b_2,\tilde{b}_2}\, \omega\bigl(-K b_2 - \tilde{K} \tilde{b}_2\bigr) \omega\bigl(2^{-1} M b_2^2 + 2^{-1} M \tilde{b}_2^2\bigr) D(M b_2, b_2) D(M \tilde{b}_2, \tilde{b}_2) \\
                                  & =  \frac{1}{p^2} \sum_{b_2,\tilde{b}_2}\, \omega\bigl(-(K b_2 + \tilde{K} \tilde{b}_2)\bigr) \omega\bigl(2^{-1} M (b_2^2 + \tilde{b}_2^2)\bigr)  \omega\bigl( M b_2 \tilde{b}_2 \bigr) D(M (b_2+\tilde{b}_2), b_2+\tilde{b}_2) \\
                                  & = \frac{1}{p^2} \sum_{b_2,\tilde{b}_2}\,  \omega\bigl(-(K b_2 + \tilde{K} \tilde{b}_2)\bigr) \omega\bigl(2^{-1} M (b_2^2 + 2 b_2 \tilde{b}_2 + \tilde{b}_2^2)\bigr) D(M (b_2+\tilde{b}_2), b_2+\tilde{b}_2) \\
                                  & = \frac{1}{p^2} \sum_{b_2,\tilde{b}_2}\, \omega\bigl(-(K b_2 + \tilde{K} \tilde{b}_2)\bigr) \omega\bigl(2^{-1} M (b_2 + \tilde{b}_2)^2\bigr) D(M (b_2+\tilde{b}_2), b_2+\tilde{b}_2) 
\end{align*}
Now to finish this proof we will replace the sum over $\tilde{b}_2$
with a sum over $B := b_2 + \tilde{b}_2$. This replacement is valid
since due to Property 1 the above sum only depends on
the value of $b_2$ and $\tilde{b}_2$ modulo $p$ and the fact that the
map $(b_2, \tilde{b}_2) \mapsto (b_2, b_2+\tilde{b}_2)$ is 
a bijection from $\Z_p$ to  $Z_p$. 
\begin{align*}
  P_{\ell_k} P_{\ell_{\tilde{k}}} & = \frac{1}{p^2} \sum_{b_2,B}\, \omega\bigl(-(K b_2 + \tilde{K} (B - b_2))\bigr) \omega\bigl(2^{-1} M B^2\bigr) D(M B, B) \\
                                  & = \frac{1}{p^2} \sum_{B}\, \omega\bigl(-\tilde{K} B\bigr) \omega\bigl(2^{-1} M B^2\bigr) D(M B, B) \sum_{b_2} \omega\bigl(-(K - \tilde{K}) b_2 \bigr).
\end{align*}
Now by Property \ref{prop2}, the sum over $b_2$ is zero if $K \neq \tilde{K}$. In the case when $K = \tilde{K}$ then we have
\[
  P_{\ell_k} P_{\ell_{\tilde{k}}} = \frac{1}{p} \sum_{B}\, \omega^{\tilde{K} B} \omega^{2^{-1} M B^2} D(M B, B) = P_{\ell_k}
\]
which proves that the collection $\{ P_{\ell_k} \}$ is a collection of mutually orthogonal projections. \\

\noindent \underline{$\sum_k P_{\ell_k} = I$:} Recall we have that
\[
  P_{\ell_k} = \frac{1}{p} \sum_{b_2}\, \omega^{K b_2 + 2^{-1} M b_2^2} D(M b_2, b_2)
\]
where $K = n^{-1} k$. Replacing $K$ with $k$ and grouping the $k$ terms we get:
\begin{align*}
  \sum_k P_{\ell_k} & = \frac{1}{p} \sum_k \sum_{b_2}\, \omega^{-n^{-1} b_2 k} \omega^{2^{-1} M b_2^2} D(M b_2, b_2) \\
                    & =\frac{1}{p} \sum_{b_2}\, \omega^{2^{-1} M b_2^2} D(M b_2, b_2)  \sum_k \omega^{-n^{-1} b_2k}
\end{align*}
Again using Property \ref{prop2}, we conclude the sum over $k$ is non-zero if and only if $b_2 = 0$. Therefore
\[
  \sum_k P_{\ell_k} = D(0,0) = I
\]
as we wanted to show.

As a remark, since $A_\alpha = W_{dis}(a_2, -a_1)$ and the mapping $(a_1, a_2) \mapsto (a_2, -a_1)$ maps lines to lines, we have also proved that Wootters' collection also satisfies his three conditions.

\section{Dynamics with Wigner and Fourier-Wigner transforms}
\label{sec:dynamics}
\subsection{The Wigner and Fourier-Wigner transforms for many-body systems}
In previous sections, we have considered the discrete Wigner and Fourier-Wigner transforms on the space $\C^p$. We can generalize this analysis to the space $(\C^p)^{\otimes N}$, the many-body case, by simply taking tensor products. In particular, if $a_1 = (a_1^1, a_1^2, \cdots, a_1^N) \in [p]^N$ and $a_2 = (a_2^1, a_2^2, \cdots, a_2^N) \in [p]^N$ then
\[
  D(a_1, a_2) := \bigotimes_{i=1}^N D(a_1^i, a_2^i).
\]
Because of this simple relationship, the many-body Heisenberg-Weyl satisfies analogous identities to the one-body case. That is $\forall a_1, a_2, b_1, b_2 \in [p]^N$:
\[
  D(a_1, a_2) D(b_1, b_2) = \omega^{a_2 \cdot b_1} D(a_1 + b_1, a_2 + b_2)
\]
\[
  D(a_1, a_2)^{-1} = D(a_1, a_2)^\dagger = \omega^{a_1 \cdot a_2} D(-a_1, -a_2).
\]
Furthermore, we can write down the corresponding many-body transforms as follows:\\

\underline{The Discrete Fourier-Wigner Transform:}
\begin{align*}
  FW_{dis}(a_1,a_2) & = \omega^{2^{-1} a_1 \cdot a_2} D(a_1, a_2) \\
                    & = \bigotimes_{i=1}^N FW_{dis}(a_1^i, a_2^i).
\end{align*}

\underline{The Discrete Wigner Transform:}
\begin{align*}
  W_{dis}(a_1,a_2) & = \cF^{-1} FW_{dis}(a_1, a_2) \\
                   & = \sum_{b_1, b_2} \omega^{-(a_1 \cdot b_1 + a_2 \cdot b_2)} FW_{dis}(b_1, b_2) \\
                   & = \bigotimes_{i=1}^N W_{dis}(a_1^i, a_2^i).
\end{align*}
It is straightforward to show that the collections $\{ W_{dis}(a_1, a_2) \}_{a_1, a_2 \in [p]^N}$ and $\{ FW_{dis}(a_1, a_2) \}_{a_1, a_2 \in [p]^N}$ form an orthogonal basis for $p^N \times p^N$ complex matrices.

\subsection{The dynamics of Wigner and Fourier-Wigner functions}
Let $\rho$ be the discrete density matrix (a positive semidefinite matrix with $\tr{\rho} = 1$). We define Wigner and
Fourier-Wigner coefficients for $\rho$ as follows:
\begin{align*}
  \rho_W(a_1, a_2) & := \tr(W_{dis}(a_1,a_2)^\dagger \rho), \\
  \rho_{FW}(a_1, a_2) & := \tr(FW_{dis}(a_1,a_2)^\dagger \rho).
\end{align*}
The orthogonality gives
\[
  \rho = \frac{1}{p^N} \sum_{a_1,a_2 \in [p]^N} \rho_W(a_1, a_2) W_{dis}(a_1,a_2)
          = \frac{1}{p^N} \sum_{a_1,a_2 \in [p]^N} \rho_{FW}(a_1, a_2) FW_{dis}(a_1,a_2).
\]
Suppose $H$ is a Hamiltonian operator, we have the von Neumann dynamics equation:
\[
  i \hbar \od{\rho}{t} = [H, \rho].
\]
We expand $H$ and $\pd{\rho}{t}$ into the Wigner and Fourier-Wigner bases in the same way we expanded $\rho$ in these bases by defining
\begin{align*}
  H_W(a_1, a_2) & := \tr(W_{dis}(a_1,a_2)^\dagger H), \\
  H_{FW}(a_1, a_2) & := \tr(FW_{dis}(a_1,a_2)^\dagger H).
\end{align*}
and similarly for $\pd{\rho}{t}$.

By performing this expansion, we get the dynamics in the Wigner basis is
\[
\frac{i \hbar}{p^N} \sum_{a_1,a_2}\pd{\rho_W(a_1,a_2)}{t} W_{dis}(a_1,a_2) = \frac{1}{(p^N)^2} \sum_{b_1,b_2} \sum_{c_1,c_2} H_W(b_1,b_2) \rho_W(c_1,c_2) [W_{dis}(b_1,b_2), W_{dis}(c_1,c_2)] 
\]
Hence
\[
i \hbar \pd{\rho_W(a_1,a_2)}{t} = \frac{1}{(p^N)^2} \sum_{b_1,b_2} \sum_{c_1,c_2} H_W(b_1,b_2) \rho_W(c_1,c_2) \tr\left(W_{dis}(a_1,a_2)^\dagger [W_{dis}(b_1,b_2), W_{dis}(c_1,c_2)] \right)
\]
Using the fact that $W_{dis}(a_1,a_2)$ is Hermitian we can define $\Gamma_{\alpha,\beta,\gamma}^W$ (where $\alpha = (a_1, a_2)$, $\beta = (b_1,b_2)$, $\gamma = (c_1,c_2)$) as follows:
\[
\Gamma_{\alpha,\beta,\gamma}^W := \tr\left(W_{dis}(a_1,a_2) W_{dis}(b_1,b_2) W_{dis}(c_1,c_2) \right)
\]
and write
\[
i \hbar \pd{\rho_W(a_1,a_2)}{t} = \frac{1}{(p^N)^2} \sum_{b_1,b_2} \sum_{c_1,c_2} (\Gamma_{\alpha, \beta, \gamma}^W - \Gamma_{\alpha, \gamma, \beta}^W) H_W(b_1,b_2) \rho_W(c_1,c_2).
\]

Performing analogous calculations for the Fourier-Wigner transform gives
\[
i \hbar \pd{\rho_{FW}(a_1,a_2)}{t} = \frac{1}{(p^N)^2} \sum_{b_1,b_2} \sum_{c_1,c_2} H_{FW}(b_1,b_2) \rho_{FW}(c_1,c_2) \tr\left(FW_{dis}(a_1,a_2)^\dagger [FW_{dis}(b_1,b_2), FW_{dis}(c_1,c_2)] \right).
\]
Noting that $FW_{dis}(a_1,a_2)^\dagger = FW_{dis}(-a_1,-a_2)$ we can define
\[
\Gamma_{\alpha,\beta,\gamma}^{FW} := \tr\left(FW_{dis}(a_1,a_2) FW_{dis}(b_1,b_2) FW_{dis}(c_1,c_2) \right) 
\]
and write
\[
i \hbar \pd{\rho_{FW}(a_1,a_2)}{t} = \frac{1}{(p^N)^2} \sum_{b_1,b_2} \sum_{c_1,c_2} (\Gamma_{-\alpha, \beta, \gamma}^{FW} - \Gamma_{-\alpha, \gamma, \beta}^{FW}) H_{FW}(b_1,b_2) \rho_{FW}(c_1,c_2).
\]
Therefore to write down the dynamics equations in the Fourier and Fourier-Wigner bases, we simply need to find an expression for $\Gamma_{\alpha,\beta,\gamma}^{W}$ and $\Gamma_{\alpha,\beta,\gamma}^{FW}$.
\subsubsection{Evaluating $\Gamma_{\alpha,\beta,\gamma}^{FW}$}
By the definition of the discrete Fourier-Wigner transform,
\[
  \begin{split}
    \Gamma_{\alpha,\beta,\gamma}^{FW} &= \tr \big( FW_{dis}(a_1, a_2) FW_{dis}(b_1,b_2) FW_{dis}(c_1,c_2) \big) \\
                                   &= \omega^{2^{-1} (a_1 \cdot a_2 + b_1 \cdot b_2 + c_1 \cdot c_2)} \tr \big(D(a_1, a_2) D(b_1, b_2) D(c_1, c_2) \big)
  \end{split}
\]
Using the properties of the multi-body Heisenberg-Weyl group we can therefore write
\[
  \Gamma_{\alpha,\beta,\gamma}^{FW} = p^N \omega^{2^{-1} (a_1 \cdot
    a_2 + b_1 \cdot b_2 + c_1 \cdot c_2)} \omega^{a_2 \cdot b_1 + (a_2
    + b_2) \cdot c_1} \, \delta^{p}_{\alpha+\beta+\gamma,0}
\]
If $p > 2$, can simplify the above equation as
\[
  \Gamma_{\alpha,\beta,\gamma}^{FW} = p^N \omega^{2^{-1} (b_2 \cdot c_1 - b_1 \cdot c_2)} \, \delta^{p}_{\alpha+\beta+\gamma,0}
\]
When $p = 2$, one can interpret $\omega^{2^{-1}}$ as $i$ and get
\[
  \Gamma_{\alpha,\beta,\gamma}^{FW} = 2^N (-1)^{a_2 \cdot b_1} i^{a_1
    \cdot a_2} i^{b_1 \cdot b_2} (-i)^{c_1 \cdot c_2} \, \delta^{p}_{\alpha+\beta+\gamma,0}
\]
\subsubsection{Evaluating $\Gamma_{\alpha,\beta,\gamma}^{W}$}
For the coefficients $\Gamma_{\alpha,\beta,\gamma}^W$, we use the relation between the Fourier-Wigner and Wigner transform to get 
\[
    \Gamma_{\alpha,\beta,\gamma}^W = \frac{1}{(p^N)^3} \sum_{a'_1,a'_2} \sum_{b'_1,b'_2} \sum_{c'_1,c'_2} \omega^{-(a_1 \cdot a_1' + a_2 \cdot a_2' + b_1 \cdot b_1' + b_2 \cdot b_2' + c_1 \cdot c_1' + c_2 \cdot c_2')} \tr \big( FW_{dis}(a_1', a_2') FW_{dis}(b_1',b_2') FW_{dis}(c_1',c_2') \big).
\]
If $p > 2$, we can use the equality
\[
  \sum_{a \in [p]^N} \omega^{a \cdot b} = p^N \delta^p_{b,0}
\]
to obtain
\[
  \begin{split}
    \Gamma_{\alpha,\beta,\gamma}^W &= \frac{1}{(p^N)^2} \sum_{b'_1,b'_2} \sum_{c'_1,c'_2} \omega^{-((b_1 - a_1) \cdot b_1' + (b_2 - a_2) \cdot b_2' + (c_1 - a_1) \cdot c_1' + (c_2 - a_2) \cdot c_2')} \omega^{2^{-1} (b_2' \cdot c_1' - b_1' \cdot c_2')} \\
                                   &= \omega^{2 (b_2 - a_2) \cdot (c_1 - a_1) - 2 (b_1 - a_1) \cdot (c_2 - a_2)}.
  \end{split}
\]
Note that when $p=2$ since $2^{-1}$ is not an integer so the above argument does not hold.

\begin{remark}
  Notice that unlike $\Gamma_{\alpha,\beta,\gamma}^{FW}$, $\Gamma_{\alpha,\beta,\gamma}^W$ is non-zero for every $\alpha,\beta,\gamma$. Recall our dynamics equation in the Wigner basis:
  \[
    i \hbar \pd{\rho_W(a_1,a_2)}{t} = \frac{1}{(p^N)^2} \sum_{b_1,b_2} \sum_{c_1,c_2} (\Gamma_{\alpha, \beta, \gamma}^W - \Gamma_{\alpha, \gamma, \beta}^W) H_W(b_1,b_2) \rho_W(c_1,c_2).
  \]
  Using the above formula for $p>2$, we can calculate that $\Gamma_{\alpha,\beta,\gamma}^W - \Gamma_{\alpha,\gamma,\beta}^W$ is not sparse (see below for $p=3$ case, $\Gamma_{-\alpha,\beta,\gamma}^{FW} - \Gamma_{-\alpha,\gamma,\beta}^{FW}$ is included for comparison):
  \[
    \renewcommand{\arraystretch}{1.4}
    \begin{array}{c|c|c}
      p^N & \mathtt{nnz}\bigl(\Gamma_{\alpha,\beta,\gamma}^W - \Gamma_{\alpha,\gamma,\beta}^W\bigr) & \mathtt{nnz}\bigl(\Gamma_{-\alpha,\beta,\gamma}^{FW} - \Gamma_{-\alpha,\gamma,\beta}^{FW}\bigr) \\
      \hline
      3 & 432 & 48 \\
      9 & 349920 & 4320 \\
    \end{array}
  \]
  While there does not seem to be a simple formula for $p=2$, we can verify that this same pattern holds numerically
    \[
    \renewcommand{\arraystretch}{1.4}
    \begin{array}{c|c|c}
      p^N & \mathtt{nnz}\bigl(\Gamma_{\alpha,\beta,\gamma}^W - \Gamma_{\alpha,\gamma,\beta}^W\bigr) & \mathtt{nnz}\bigl(\Gamma_{-\alpha,\beta,\gamma}^{FW} - \Gamma_{-\alpha,\gamma,\beta}^{FW}\bigr) \\
      \hline
      2 & 24 & 6 \\
      4 & 2208 & 120 \\
    \end{array}
  \]
  We note that this difference in the sparsity of $\Gamma_{\alpha,\beta,\gamma}^W - \Gamma_{\alpha,\gamma,\beta}^W$ and $\Gamma_{\alpha,\beta,\gamma}^{FW} - \Gamma_{\alpha,\gamma,\beta}^{FW}$ may drastically increase the number of computations needed when performing numerical approximations.
\end{remark}

\subsection{Dynamics of the many-body spin system in phase space representation}
Using the expressions for $\Gamma^{FW}_{\alpha,\beta,\gamma}$ we found
above, we can easily expand the dynamics for any two spin
system. Throughout this section we will assume $p=2$ and write
$e_i \in \Z^N_2$ as the $i^{\text{th}}$ standard basis vector of
$\Z^N_2$. As in \cite{Schachenmayer2015}, we consider the following
Hamiltonian:
\[
  H = \frac{1}{2} \sum_{i\neq j} \left[ \frac{J_{ij}^{\perp}}{2} (X_i X_j + Y_i Y_j) + J_{ij}^z Z_i Z_j \right] + \Omega \sum_i X_i.
\]
From our calculations above we get,
\[
  H_{FW}(a_1, a_2) = \left\{ \begin{array}{ll}
                            \frac{p^N}{4} J_{ij}^{\perp}, & \text{if } a_1 = e_i + e_j, a_2 = 0, i \neq j, \\
                            \frac{p^N}{4} J_{ij}^{\perp}, & \text{if } a_1 = a_2 = e_i + e_j, i \neq j, \\
                            \frac{p^N}{2} J_{ij}^z, & \text{if } a_1 = 0, a_2 = e_i + e_j, i \neq j, \\
                            p^N \Omega, & \text{if } a_1 = e_i, a_2 = 0, \\
			    0, & \text{otherwise}.
                          \end{array} \right.
\]
Therefore, we arrive at 
\begin{equation} \label{eq:spin_dynamics}
  \begin{split}
  & i\hbar \pd{\rho_{FW}(a_1,a_2)}{t} = i^{a_1 \cdot a_2} \Bigg( \frac{J_{ij}^{\perp}}{4} \sum_{i\neq j} i^{a_2 \cdot (a_1 + e_i + e_j)} [(-1)^{a_1 \cdot a_2} - (-1)^{a_2 \cdot (a_1 + e_i + e_j)}] \rho_{FW}(a_1 + e_i + e_j, a_2) \\
				   & \quad + \frac{J_{ij}^{\perp}}{4} \sum_{i\neq j} i^{(a_1 + e_i + e_j) \cdot (a_2 + e_i + e_j)} [(-1)^{a_2 \cdot (a_1 + e_i + e_j)} - (-1)^{a_1 \cdot (a_2 + e_i + e_j)}] \rho_{FW}(a_1 + e_i + e_j, a_2 + e_i + e_j) \\
                                   & \quad + \frac{J_{ij}^z}{2} \sum_{i\neq j} i^{a_1 \cdot (a_2 + e_i + e_j)} [(-1)^{a_1 \cdot (a_2 + e_i + e_j)} - (-1)^{a_1 \cdot a_2}] \rho_{FW}(a_1, a_2 + e_i + e_j) \\
                                   & \quad + \Omega \sum_i i^{a_2 \cdot (a_1 + e_i)} [(-1)^{a_1 \cdot a_2} - (-1)^{a_2 \cdot (a_1 + e_i)}] \rho_{FW}(a_1 + e_i, a_2) \Bigg).
  \end{split}
\end{equation}

Specializing our equations above for the two-spin system (i.e. $N=2$) the dynamics is

\[
  \renewcommand*{\arraystretch}{1.3}
  \begin{split}
  & \od{}{t} \begin{pmatrix}
    \rho_{FW}^{00,00} & \rho_{FW}^{00,01} & \rho_{FW}^{00,10} & \rho_{FW}^{00,11} \\
    \rho_{FW}^{01,00} & \rho_{FW}^{01,01} & \rho_{FW}^{01,10} & \rho_{FW}^{01,11} \\
    \rho_{FW}^{10,00} & \rho_{FW}^{10,01} & \rho_{FW}^{10,10} & \rho_{FW}^{10,11} \\
    \rho_{FW}^{11,00} & \rho_{FW}^{11,01} & \rho_{FW}^{11,10} & \rho_{FW}^{11,11}
  \end{pmatrix} = J_{12}^{\perp} \begin{pmatrix}
    0 & \rho_{FW}^{11,01} - \rho_{FW}^{11,10} & \rho_{FW}^{11,10} - \rho_{FW}^{11,01} & 0 \\
    \rho_{FW}^{10,11} & -\rho_{FW}^{10,01} & \rho_{FW}^{10,10} & -\rho_{FW}^{10,00} \\
    \rho_{FW}^{01,11} & \rho_{FW}^{01,01} & -\rho_{FW}^{01,10} & -\rho_{FW}^{01,00} \\
    0 & \rho_{FW}^{00,01} - \rho_{FW}^{00,10} & \rho_{FW}^{00,10} - \rho_{FW}^{00,01} & 0
  \end{pmatrix} \\
  & \quad + 2J_{ij}^z \begin{pmatrix}
    0 & 0 & 0 & 0 \\
    -\rho_{FW}^{01,11} & \rho_{FW}^{01,10} & -\rho_{FW}^{01,01} & \rho_{FW}^{01,00} \\
    -\rho_{FW}^{10,11} & -\rho_{FW}^{10,10} & \rho_{FW}^{10,01} & \rho_{FW}^{10,00} \\
    0 & 0 & 0 & 0
  \end{pmatrix} + 2\Omega \begin{pmatrix}
    0 & \rho_{FW}^{01,01} & \rho_{FW}^{10,10} & \rho_{FW}^{10,11} + \rho_{FW}^{01,11} \\
    0 & -\rho_{FW}^{00,01} & \rho_{FW}^{11,10} & \rho_{FW}^{11,11} - \rho_{FW}^{00,11} \\
    0 & \rho_{FW}^{11,01} & -\rho_{FW}^{00,10} & \rho_{FW}^{11,11} - \rho_{FW}^{00,11} \\
    0 & -\rho_{FW}^{10,01} & -\rho_{FW}^{01,10} & -(\rho_{FW}^{01,11} + \rho_{FW}^{10,11})
  \end{pmatrix},
  \end{split}
\]
where $\rho_{FW}^{i_1 i_2, j_1 j_2}$ stands for $\rho_{FW}((i_1, i_2), (j_1, j_2))$ and we have chosen units so that $\hbar=1$.

\subsection{Tensor-product ansatz}
In the many-body spin system, when the entanglement between spins is not strong, the state of the system can be approximated by the tensor-product ansatz. The corresponding density matrix of this ansatz is
\[
  \rho = \rho^{(1)} \otimes \rho^{(2)} \otimes \cdots \otimes \rho^{(N)},
\]
where each $\rho^{(k)}$ is a one-body density matrix satisfying
$\tr(\rho^{(k)}) \equiv 1$. To evolve this ansatz for the density
matrix, the simplest method is to ignore all the coupling terms in the
Hamiltonian, i.e. $H \approx H_0 = \Omega \sum_i X_i$. Such an idea
was used in the quantum kinetic Monte Carlo method proposed in
\cite{Cai2018} by two of the authors, where the Dyson series expansion
is used in the computation and its leading order term is just the
tensor-product state evolved with the Hamiltonian $H_0$. From the
perspective of the Fourier and Fourier-Wigner transform, the ansatz
turns out to be
\begin{equation} \label{eq:tensor_product}
  \begin{split}
    \rho_W(a_1, a_2) = \prod_{k=1}^N \rho_W \big(a_1^{(k)} e_k, a_2^{(k)} e_k\big), & \qquad \rho_{FW}(a_1, a_2) = \prod_{k=1}^N \rho_{FW} \big(a_1^{(k)} e_k, a_2^{(k)} e_k\big), \\
    & \forall a_1 = (a_1^{(1)}, \cdots, a_1^{(N)}) \in \Z_2^N, \quad a_2 = (a_2^{(1)}, \cdots, a_2^{(N)}) \in \Z_2^N.
  \end{split}
\end{equation}
The dynamics of $\rho_{FW}$ with Hamiltonian $H_0$, viewed as a
Galerkin projection to the tensor product ansatz, can be obtained by
removing all the terms with $J_{ij}^{\perp}$ and $J_{ij}^z$ from
\eqref{eq:spin_dynamics}.

Another method to evolve the density matrix $\rho$ is proposed in
\cite{Schachenmayer2015}, where the tensor-product ansatz is linked to
a classical interacting spin system. Using the development in this
paper, we can now explicitly formulate such a connection. According to
the ansatz \eqref{eq:tensor_product}, the dynamics of the
Fourier-Wigner transform is fully determined by the evolution of
$\rho_{FW}(0,e_k)$, $\rho_{FW}(e_k,0)$ and $\rho_{FW}(e_k, e_k)$ for
$k = 1,\cdots,N$.  Assuming $\hbar = 1$, we can write down the
equations for these three coefficients using \eqref{eq:spin_dynamics}
as
\begin{align*}
  \od{\rho_{FW}(e_k,0)}{t} &= J_{ij}^{\perp} \rho_{FW}(0,e_k) \sum_{i\neq k} \rho(e_i,e_i) - 2 J_{ij}^z \rho_{FW}(e_k,e_k) \sum_{i\neq k} \rho_{FW}(0,e_i), \\
  \od{\rho_{FW}(0,e_k)}{t} &= J_{ij}^{\perp} \rho_{FW}(e_k,e_k) \sum_{i\neq k} \rho_{FW}(e_i,0) - J_{ij}^{\perp} \rho_{FW}(e_k,0) \sum_{i\neq k} \rho_{FW}(e_i,e_i) + 2\Omega \rho_{FW}(e_k,e_k), \\
  \od{\rho_{FW}(e_k,e_k)}{t} &= -J_{ij}^{\perp} \rho_{FW}(0,e_k) \sum_{i\neq k} \rho_{FW}(e_i,0) + 2 J_{ij}^z \rho_{FW}(e_k,0) \sum_{i\neq k} \rho_{FW}(0,e_i) - 2\Omega \rho_{FW}(0,e_k).
\end{align*}
Here the ansatz \eqref{eq:tensor_product} has been applied to write
all quantities in terms of $\rho_{FW}(0,e_k)$, $\rho_{FW}(e_k,0)$ and
$\rho_{FW}(e_k, e_k)$.  By interpreting $\rho_{FW}(e_k,0)$,
$\rho_{FW}(e_k,e_k)$ and $\rho_{FW}(0,e_k)$ as quantities proportional
to the $x$, $y$, and $z$ angular momentum of the $k$-th classical
spin, the above system is exactly the classical spin system introduced
in \cite{Schachenmayer2015, Krech1998}. The numerical analysis of such
tensor product ansatz, which involves understanding of propagation of
quantum entanglement, will be left for future works.

\section*{Acknowledgements}
This work is partially supported by the National Science
  Foundation under award RNMS11-07444 (KI-Net). ZC is also supported
  by National University of Singapore Startup Fund under Grant
  No.~R-146-000-241-133, JL is supported in part by National Science
  Foundation award DMS-1454939, and KS is also supported by a National
  Science Foundation Graduate Research Fellowship under Grant No.
  DGE-1644868. We also thank Robert Calderbank for helpful discussions.

\bibliographystyle{plain}
\bibliography{discrete_wigner_final}

\end{document}